\documentclass[fleqn,12pt,twoside]{article}

\usepackage{graphicx}
\usepackage[figuresright]{rotating}
\usepackage{setspace}

\newtheorem{theorem}{Theorem}
\newtheorem{lemma}{Lemma}
\newtheorem{corollary}{Corollary}
\newtheorem{claim}{Claim}

\newtheorem{proposition}{Proposition}

\usepackage{espcrc1}

\newtheorem{problem}{Problem}
\newtheorem{remark}{Remark}

\newenvironment{proof}[1][Proof.]{\begin{trivlist}
\item[\hskip \labelsep {\bfseries #1}]}{\end{trivlist}}

\usepackage{amsmath}
\usepackage{amsfonts}
\begin{document}

\title{Measures of edge-uncolorability}

\author{Vahan V. Mkrtchyan\address[MCSD]{Paderborn Institute for Advanced Studies in Computer Science and Engineering,
Paderborn University, Warburger Str. 100, 33098 Paderborn, Germany}%
\thanks{The author is supported by a fellowship from Heinrich Hertz-Stiftung}
\thanks{email: vahanmkrtchyan2002@\{ysu.am, ipia.sci.am,
yahoo.com\}}
                        and
        Eckhard Steffen \addressmark[MCSD]\thanks{email: es@upb.de}}
%\thanks{email: vahanmkrtchyan2002@yahoo.com, es@upb.de}
% If you use the option headings,
% the title is also used as the running title,
% and the authors are also used as the running authors.
% You can change that by using \runtitle and \runauthor.

%\runtitle{The oddness of regular graphs}
%\runauthor{Vahan Mkrtchyan, Eckhard Steffen}

\maketitle

\begin{abstract} The resistance $r(G)$ of a graph $G$ is the minimum number of edges that have to be removed from $G$ to obtain a graph which
is $\Delta(G)$-edge-colorable.
The paper relates the resistance to other parameters that measure how far is a graph from being $\Delta$-edge-colorable. The first part considers regular graphs
and the relation of the resistance to structural properties in terms of 2-factors. The second part studies general (multi-) graphs $G$. 
Let $r_v(G)$ be the minimum number of vertices that have to be removed from $G$ to obtain a class 1 graph. We show that
$\frac{r(G)}{r_v(G)} \leq \lfloor \frac{\Delta(G)}{2} \rfloor$, and that this bound is best possible.
\end{abstract}

\section{Introduction and definitions}

We consider finite (multi-) graphs $G=(V,E)$ with vertex set $V$ and edge set $E$. The graphs might have multiple 
edges but no loops, and throughout this paper we assume that the graphs under consideration are connected (otherwise we will mention this explicitly). 
Terms and concepts that we do not define can be found in \cite{Lov}. A $k$-edge-coloring of a graph $G=(V,E)$ is a function $\phi$ from the edges of $G$ to a set $\cal{C}$ of $k$ colors, 
such that
adjacent edges receive different colors. The minimum number $k$ such that $G$ has a $k$-edge-coloring is the chromatic index of $G$, and it is denoted by $\chi'(G)$. Clearly, $\Delta(G) \leq \chi'(G)$. If $\chi'(G)=\Delta(G)$, then $G$ is class 1, 
otherwise it is class 2. 

Let $G$ be a (multi-) graph. The {\em resistance} $r(G)$ of $G$ is the minimum number of edges that have to be removed from $G$ 
to obtain a $\Delta(G)$-edge colorable subgraph $H$. In \cite{classIIgraphs} such subgraphs $H$ are called maximum $\Delta(G)$-edge-colorable subgraphs, and it is shown that $\Delta(H) \geq \frac23 \Delta(G)$ for any maximum $\Delta(G)$-edge-colorable subgraph $H$ of $G$, and if $G$ is simple, then $\Delta(H) = \Delta(G)$.

Furthermore, let $r'_{v}$ denote the minimum number of vertices that have to be removed from $G$
to obtain a graph $H$ with $\chi'(H) \leq \Delta(G)$. A modification of this parameter is 
$r_v(G)$ which denotes the minimum number of vertices that have to be removed
from $G$ to obtain a class 1 graph.     

The {\em oddness} $\xi_3(G)$ of a cubic graph $G$ is the minimum number of odd cycles in a 2-factor of $G$, if $G$ possesses a 2-factor, and is infinite, otherwise. 
This notion was explicitly used first by Huck and Kochol in \cite{Huck_Kochol_1995} to prove the 
double-cycle-cover-conjecture (CDCC) for 
cubic graphs with oddness 2. Implicitly it appears already in Jaeger's paper \cite{Jaeger_1988}, when he proved the 5-flow conjecture for graphs with oddness 2. 
This parameter gives us some information about how far the graph is from being class 1. Clearly $G$ is 3-edge-colorable if and only if 
$\xi_3(G) = 0$ and $\xi_3(G)$ is always even. There are some applications of this parameter. For instance, in 
\cite{Huck_2001} the CDCC is proved for cubic graphs with oddness 4, and in \cite{Steffen_2009} it is shown that cyclically $k$-edge connected graphs have a nowhere-zero 5-flow 
if $k \geq \frac{5}{2} \xi_3 - 3$.

This paper generalizes the notion of oddness to regular graphs. Recall that a subgraph of a graph is called spanning, if the vertex set of the subgraph coincides with that of the graph. A subgraph of a graph is called a $k$-factor, if it is spanning and $k$-regular. Petersen \cite{Petersen_1891} showed that every $2n$-regular graph $G$ ($n \geq 0$) has a 2-factorization, that is, a decomposition of $E(G)$ into $n$ 2-factors. 
Let $G$ be a $2n$-regular graph, and ${\cal F}_2^*(G)$ be the set of all 2-factorizations of $G$. Let 
${\cal F} \in {\cal F}_2^*(G)$ then $o({\cal F})$ denotes the number of odd cycles in ${\cal F}$; note that ${\cal F}$ consists of $n$ edge disjoint 2-factors of $G$.

Let $G$ be an $s$-regular graph. We define the {\em oddness} of $G$, to be 

\[ \xi(G) = \min\{o(F)| F \in {\cal F}_2^*(G)\} , \textrm{if } s \textrm{ is even, and }\]

\[ \xi(G) = \min\{\xi(G-F_1)| F_1   \textrm{ is a 1-factor of } G\},  \textrm{if } s \textrm{ is odd, and } $G$ \textrm{ has a 1-factor,}\] 

and let us define it to be infinite, otherwise.

Let $G$ be an $s$-regular graph. A 2-factor of $G$ with odd (only even) cycles is called an odd (even) 2-factor.
Let ${\cal F}$ be a 2-factorization of $G$, then every 2-factor has a 2-coloring that leaves precisely one edge uncolored in each odd cycle. 
Such an edge-coloring is called a {\em canonical} coloring of $G$. (If $s$ is odd one has to add one 1-factor.)

Following \cite{Steffen_2004} the {\em resistance} of an $s$-regular graphs is denoted by $r(G)$.  

In the second section we study regular (multi-) graphs and the relation between the resistance and the oddness. We show that for every $n$ 
there is an $s$-regular graph $G$ such that $r(G) = \xi(G)$, and that these parameters can be arbitrary far apart even for $s$-graphs. The concept of an $s$-graph was introduced by Seymour in \cite{Seymour_1979} as a generalization of that of a bridgeless cubic graph. Thus, our results can be considered as a generalization of the corresponding results obtained in \cite{Steffen_2004}.

In the third section we study general (multi-) graphs $G$. Clearly, $r'_v(G)\leq r_v(G)$, and we show that the difference between these parameters can be arbitrary big. We also show that
$\frac{r(G)}{r_v(G)}, \frac{r(G)}{r'_{v}(G)} \leq \lfloor \frac{\Delta(G)}{2} \rfloor$, and that these bounds are best possible. While in general case $r_v(G)$ and $r'_v(G)$ are different, it can be easily seen that they are the same for the case of cubic graphs, and the last inequalities can be considered as a generalization of the corresponding results for cubic graphs that was first proved in \cite{Steffen_1998}.

\section{Regular graphs: Oddness and resistance}

Throughout the paper we will use the following well-known generalization of the Parity Lemma. For the sake of completeness we add a short proof. 

\begin{lemma} \label{Parity lemma}
 Let $G$ be a graph whose edges are colored with colors $1, \dots , c$, and let $a_i$ be the number of
vertices $v$ in $G$ such that no edge incident to $v$ is colored $i$. Then for all $i = 1, \dots c:$ 
$a_i \equiv |V(G)|\pmod2$.
\end{lemma}

{\bf Proof.}  For $i = 1, \dots , c$ let $E_i$ be the set of edges which are colored with color $i$. Then $a_i = |V(M)| - 2|E_i|$, and hence 
$a_i \equiv |V(G)|\pmod2$.\hfill\ {\raisebox{0.8ex}{\framebox{}}}\par\bigskip

Clearly, the oddness can  be odd only for graphs of odd order. The first part of this section proves results
which are true independent from the order of the graphs.

\begin{lemma}\label{class1case} For every $s$-regular graph $G$, the following statements are equivalent:
\begin{enumerate}
 \item $r(G) = 0$.
 \item $\xi(G)=0$.
 \item $\chi'(G)=s$.
\end{enumerate}
\end{lemma}

{\bf Proof.}
(1 $\Rightarrow$ 2, 3) Since $G$ is regular, $r(G) = 0$ implies that $G$ is of even order and $\chi'(G)=s$. Every color class is a 1-factor of $G$ and hence $\xi(G)=0$.

(2 $\Rightarrow$ 1)
$\xi(G) = 0$ implies that $E(G)$ can be partitioned into edge-disjoint 1-factors. Hence $\chi'(G)=s$. 

(3 $\Rightarrow$ 1)
$\chi'(G)=s$ means that the graph $G$ is $s$-edge-colorable, hence the minimum number of edges that should be removed from $G$ in order to obtain an $s$-edge-colorable subgraph, is zero, that is, $r(G) = 0$.\hfill\ {\raisebox{0.8ex}{\framebox{}}}\par\bigskip

\begin{lemma} \label{upper_bound} For every $s$-regular graph $G$:  
$r(G)\leq \xi(G)$.
\end{lemma}

{\bf Proof.} The cases $s=0,1$ are trivial, since $r(G)=0$, and $\xi(G)$ is infinite. For $s\geq 2$ the proof will be given by induction on $s$. The case $s = 2$ is trivial again, while the case $s = 3$ is shown in \cite{Steffen_2004}. Now, choose a $2$-factor $F$ with $k$ odd cycles in a minimum $2$-factorization of $G$, that is, in a $2$-factorization of $G$ that contains $\xi(G)$ odd cycles. Adding $F$ to $G-F$ and extending a $(s-2)$-coloring of $G-F$ with at most $\xi(G)-k$ uncolored edges yields an $s$-coloring of 
$G$ with at most $\xi(G)$ uncolored edges. By induction, we have: $r(G-F)\leq \xi(G-F)=\xi(G)-k$. Therefore $r(G)\leq \xi(G)$.
\hfill\ {\raisebox{0.8ex}{\framebox{}}}\par\bigskip

\begin{lemma} \label{one_odd_2factor}
If there is a $2$-factor containing all $\xi(G)$ odd cycles of a $2$-factorization of an $s$-regular graph, 
then $\chi'(G)\leq s+1$.
\end{lemma}

{\bf Proof.} Suppose there is 2-factor $F$ containing all $\xi(G)$ odd cycles of a 2-factorization of $G$. Then the edges of the 2-factor $F$ can be colored with three colors and the edges of each of the remaining 2-factors in $\cal{F}$ can be colored with two (fresh) colors. Hence $\chi'(G)\leq3+2(|\cal F|$ $-1)=s+1$. 

\hfill\ {\raisebox{0.8ex}{\framebox{}}}\par\bigskip

\subsection{Regular graphs of odd order}

\vspace{.2cm}

This section considers $s$-regular graphs of odd order, i.e. $s$ is even.

\begin{proposition} \label{Prop_odd_1}
 Let $G$ be an $s$-regular graph of odd order, then $r(G) \geq \frac{s}{2}$  and $\xi(G) \geq \frac{s}{2}$.
\end{proposition}

{\bf Proof.} Let $2n+1$ ($n \geq 1$) be the order of $G$. Then $s$ is even, say $s=2k$ ($k \geq 1$), and $|E(G)| = 2nk + k$, and every color class has at most $n$ edges. Thus 
$r(G) \geq \frac{s}{2}$. Lemma \ref{upper_bound} implies that $\xi(G) \geq \frac{s}{2}$.
\hfill\ {\raisebox{0.8ex}{\framebox{}}}\par\bigskip

\begin{proposition}\label{s2proposition}
  Let $G$ be an $s$-regular graph of odd order. Then the following statements are equivalent:
\begin{enumerate}
 \item $r(G) = \frac{s}{2}$.
 \item $\xi(G) = \frac{s}{2}$.
 \item The graph $G$ has a 2-factorization $\cal F$ in which every 2-factor $F$ contains precisely one odd cycle.
\end{enumerate}
\end{proposition}

{\bf Proof.} Let $G$ be of order $2n+1$ ($n \geq 1$) and $s = 2k$ ($k \geq 1$).

(1 $\Rightarrow$ 2)  Let $c$ be a minimum $2k$-coloring of $G$ with colors $1,2, \dots , 2k$. Let
$e_1 = v_1v_2, \dots e_k=v_{2k-1}v_{2k}$ be the uncolored edges. Since $G$ is of odd order, it follows from  Lemma \ref{Parity lemma} that every color is missing at precisely one vertex, say color $i$ is missing at vertex $v_i$. Then
for every $1\leq i \leq k$, the union of $c^{-1}(2i-1)$, $c^{-1}(2i)$ and edge $e_i$ forms a 2-factor of $G$ with precisely 
one odd cycle. Therefore $\xi(G) \leq \frac{s}{2}$, and by Proposition \ref{Prop_odd_1} it follows that $\xi(G) = \frac{s}{2}$.

(2 $\Rightarrow$ 3) Since every 2-factorization has $k$ 2-factors, and each of them has at least one odd cycle, it follows that 
each 2-factor has precisely one odd cycle. 

(3 $\Rightarrow$ 1) 
This 2-factorization allows an $2k$-coloring of $G$ that leaves $k$ edges uncolored. Now the statement follows with 
Proposition \ref{Prop_odd_1}.
\hfill\ {\raisebox{0.8ex}{\framebox{}}}\par\bigskip

By Petersen's Theorem, every complete graph of odd order has a 2-factorization, and so proposition \ref{s2proposition} implies the following corollary:

\begin{corollary}
 For the complete graphs of odd order: $r_{2k}(K_{2k+1}) = \xi_{2k}(K_{2k+1}) = k$, $k \geq 1$.
\end{corollary}

\subsection{Regular graphs of even order and $s$-graphs}

\vspace{.2cm}

\begin{lemma} \label{not1_and_basic2}
Let $G$  be an $s$-regular graph of even order, then:

\begin{enumerate}
 \item $r(G) \not = 1$, and 
 \item $r(G) = 2$ if and only if $\xi(G) = 2$. 
\end{enumerate}
\end{lemma}

{\bf Proof.}
1. Assume that there is an $s$-coloring of $G$, which leaves precisely one edge $e=vw$ uncolored, i.e. $v$ and $w$ are the only vertices where colors are missing, 
say color $\alpha $ at $v$ and $\beta$ at $w$. 
Lemma \ref{Parity lemma} implies $\alpha = \beta$, and hence $G$ is $s$-colorable, i.e. $r(G) =0$, a contradiction.

2. Suppose that $\xi(G) = 2$. Then by (1) and Lemma \ref{upper_bound}, $r(G)\in \{0,2\}$. If $r(G)=0$, then $\xi(G) =0$ by Lemma \ref{class1case}. Thus $r(G)=2$.

Suppose $r(G) = 2$. Then, by Lemma \ref{upper_bound}, $\xi(G) \geq 2$. Let $c$ denote an $s$-coloring of $G$ that leaves exactly two edges uncolored, and let $e=vw$ and $f=xy$ denote the two edges of $G$ which are not colored by $c$. If $e$ and $f$ are adjacent, say $v=x$, then two different colors $\alpha$ and $\beta$ are missing at $v$. Lemma \ref{Parity lemma}
implies that, say $\alpha$, is missing at $w$. Thus $r(G) < 2$, contradicting the choice of $c$. 

Thus we can assume that $e$ and $f$ are not adjacent. Let $\alpha$ and $\beta$ be the colors missing at vertices $v$ and $w$, respectively. The choice of $c$ implies that $\alpha \not = \beta$. 
It follows with Lemma \ref{Parity lemma} that $\alpha$ and $\beta$ are missing at precisely one vertex of $f$. Therefore $c^{-1}(\alpha)$ and $c^{-1}(\beta)$ together with $e$ and $f$ induce a 2-factor of $G$ with precisely two odd cycles. All other color classes are 1-factors and therefore  $\xi(G) = 2$.
\hfill\ {\raisebox{0.8ex}{\framebox{}}}\par\bigskip

An $s$-graph is an $s$-regular graph $G$ such that $|\partial_G(X)| \geq s$ for every odd set $X \subseteq V(G)$, where $\partial_G(X)$ is the set of edges of $G$ connecting $X$ to $V(G)\backslash X$.

The following proposition will be used frequently:

\begin{proposition}\label{sgraphprop}Let $G$ be an $s$-graph ($s\geq 1$). Then:
\begin{enumerate}
 \item the graph $G$ has an even number of vertices;
 \item the oddness $\xi(G)$ is even;
 \item if $G$ is connected, then it is $2$-connected, unless $G$ is the trivial graph $K_2$. 
\end{enumerate}
\end{proposition}

For $i= 0, \dots ,n$ let $G_i$ be graphs, $v_{i}w_{i} = e_i \in E(G_i)$, and let
$G = \sum_{i=0}^{n} G_i$ be the graph obtained from $G_i-e_i$ by adding edges as follows:
If $n=2k$, then add edges $v_{2j}v_{2j+1}$, $w_{2j+1}w_{2j+2}$ ($0 \leq j \leq k-1$) and $v_{2k}w_0$, and 
if $n=2k-1$, then add edges $v_{2j}v_{2j+1}$, $w_{2j+1}w_{2j+2}$ ($0 \leq j \leq k-1$)
where the indices are added modulo $2k$. 

\begin{lemma} \label{s_graph_construction}
For $i= 1, \dots ,n$, let $G_i$ be $s$-graphs, then $G = \sum_{i=1}^n G_i$ is an $s$-graph.
\end{lemma}

{\bf Proof.}
× $G$ is $s$-regular and of even order. We have to show that $|\partial_G(X)| \geq s$
for any odd set $X \subseteq V(G)$. We may assume that $G[X]$ is connected, and let
$X_i = X \cap V(G_i)$ ($i=1, \dots ,n$). Since $|X|$ is odd it follows that there is 
$1 \leq j \leq n$ such that $|X_j|$ is odd. 

If $X = X_j$ then - independent from whether $v_j, w_j \in X$ - it follows that
$|\partial_G(X)| \geq |\partial_{G_j}(X)| \geq s$.

In any other case we have that $|\partial_G(X_j)| \geq |\partial_{G_j}(X_j)| - 1 \geq s-1$. Since $G$
is 2-connected it follows that $|\partial_G(X_j)| \geq s$.
\hfill\ {\raisebox{0.8ex}{\framebox{}}}\par\bigskip

\begin{lemma} \label{basic_module}
For each positive integer $k$, there is an $s$-graph $O_k$ ($s\geq 3$)
such that $r(O_k) = 2k+1 = \xi(O_k)-1$.
\end{lemma}

{\bf Proof.} Let $M$ be a 1-factor of the Petersen graph $P$. The Meredith graph $M_s$ is a graph which is obtained from
$P$ by adding $s-3$ copies of $M$. It can be easily seen that the Meredith graph $M_s$ is an $s$-graph, and it is unique up to isomorphism. Meredith \cite{Meredith_1973} showed that  $\chi'(M_s) = s+1$. 

Since $\chi'(M_s) = s+1$, it follows easily from Lemma \ref{not1_and_basic2} that $r(M_s) =2$, and so also $\xi(M_s) = 2$. Also note that $\chi'(M_s) = s+1$ and $r(M_s) =2$ imply that $\chi'(M_s - e) = s+1$ for each edge $e$.

Take $2k+1$ copies $H_0 - e_0, \dots ,H_{2k} - e_{2k}$ of the Meredith graph $M_s - e$, where $e=vw$ is a simple edge and $e_i=v_iw_i$ for all $i\in [2k]$. 

Let $c_i$ be an $s$-coloring of 
$H_i - e_i$ that leaves exactly $r(H_i - e_i)=1$ edges uncolored, and such that the color $1$ is missing at $w_{0}$ and $v_{i}$ for $i=2, \dots, 2k$, color $3$ is
missing at $v_0$ and $v_1$, and color $2$ is missing at $w_i$ for $i=1, \dots 2k$. 

Such colorings of $M_s-e_i$ exist and it is possible to extend such a coloring to an $s$-coloring of $O_k = \sum_{i=0}^{2k} H_i$. Hence 
$r(O_k) \leq 2k+1$ and since $O_k$ contains $2k+1$ pairwise disjoint copies of $M_s - e$, it follows that
$r(O_k) = 2k+1$.

On the other hand it is easy to see that there is a 2-factor of $O_k$ which consists of a cycle $C$ of
length $5(2k+1)$, and $v_i,w_i \in V(C)$ for all $0 \leq i \leq 2k$, and of $2k+1$ cycles $C_0, \dots C_{2k}$ of length 5, and $C_i$ covers the 5 vertices in $H_i$ which are not in $C$. Since $\xi(O_k)$ is even and greater or equal to 
$r(O_k) = 2k+1$
the statement follows.
\hfill\ {\raisebox{0.8ex}{\framebox{}}}\par\bigskip

The construction above with an even number of copies of $M_s-e$ yields
an $s$-graph $G$ with $r(G) = k = \xi(G)$ for every even number $k$. For 2-regular graphs we have $r_2 = \xi_2$, and 
the oddness is an upper bound for the resistance of a graph by Lemma \ref{upper_bound}.
It is a natural question, whether the oddness of a graph can be bounded in terms of its resistance. 
However, the following theorem and corollaries show that resistance and oddness can be arbitrary far apart in $s$-graphs. \\

\begin{theorem} \label{no_constant_bound}For all $s\geq 3$, if there is a constant $t>1$ and an $s$-graph $H$ with $r(H)>2$ and $t(r(H)-2) \leq \xi(H) -2$, then there is no constant $c$ such that $1 \leq c < t$ and $\xi(G) \leq cr(G)$, for all $s$-graphs $G$.
\end{theorem}

{\bf Proof.} We show that there is a sequence of graphs $G_1, G_2 \dots$, such that for any $c < t$ there is a $k$ such that $\xi(G_k) \geq c r(G_k)$. 

Let $c$ be an $s$-coloring of $H$ that leaves $r(H)$ edges uncolored. Let $G_1=H$, then $\xi(G_1) \geq t(r(G_1) - 2) +2$ by our assumption.

Let $k\geq 2$, and suppose that we have constructed $s$-graphs $G_1,...,G_{k-1}$ with $r(G_j)=jr(G_1)-2(j-1)$ and $\xi(G_j)\geq j\xi(G_1)-2(j-1)$ for all $j\in [k-1]$.

Suppose that $c_{k-1}$ and $c_1$ are $s$-colorings of $G_{k-1}$ and $G_1$ that leave $r(G_{k-1})$ and $r(G_1)$ edges uncolored, respectively. We may assume - by appropriate labeling of the colors - that there is an uncolored edge $vw$ in $G_{k-1}$ and an uncolored edge $xy$ in $G_1$ where the same colors are missing, say color 1 at $v$ and $x$ and color 2 at $w$ and $y$. 

Let $G_k$ be the graph obtained from $G_{k-1} - vw$ and $G_1 - xy$ by adding the edges $e=vx$ and $f=wy$. $G_k$ is an 
$s$-graph by Lemma \ref{s_graph_construction}, and the colorings $c_{k-1}$ and $c_1$ can be extended
to a coloring of $G_k$ with $r(G_{k-1}) + r(G_1) - 2$ uncolored edges. 

Assume that there is a coloring of $G_k$ with less than $r(G_{k-1}) + r(G_1) - 2$ uncolored edges. This implies that $G_k[V(G_1)]$ or $G_k[V(G_{k-1})]$ contains 
at most $r(G_1) - 2$  or at most $r(G_{k-1}) - 2$ uncolored edges, respectively. Thus we can obtain a coloring of $G_1$ or of $G_{k-1}$ with at most $r(G_1) - 1$ or  with at most $r(G_{k-1}) - 1$ uncolored edges, respectively, which is a contradiction. Thus
$r(G_k) = r(G_{k-1}) + r(G_1) - 2 = kr(G_1) - 2(k-1) = kr(H) - 2(k-1)$.

We will show that $\xi(G_k) \geq k\xi(G_1) - 2(k-1)$. 
Let ${\cal F}$ be a minimum 2-factorization of $G$, and recall that $\{e,f\}$ is a 2-cut in $G_k$. 

If $e$ and $f$ do not belong to any 2-factor (i.e. $s$ is odd), then, since ${\cal F}$ is a minimum 2-factorization,
$
\xi(G_k) 
\geq  \xi(G_{k-1}) + \xi(G_1) 
\geq (k-1)\xi(G_1) - 2(k-2) + \xi(G_1)
> k\xi(G_1) - 2(k-1) = k\xi(H) - 2(k-1)$.

If $e$ and $f$ are edges of a 2-factor of $G_k$, then 
there is one 2-factor $F$ of ${\cal F}$ and one cycle $C$ in $F$ with $e,f \in E(C)$.
Thus $F$ induces a 2-factor in $G_{k-1}$ and $G_1$ and it adds at most one odd cycle to the odd cycles of 
${\cal F}$ in $G[V(G_{k-1})]$ and in $G[V(G_1)]$. Furthermore, these are 2-factorization of $G_{k-1}$ and $G_1$. 
Thus $\xi(G_k) 
\geq  \xi(G_{k-1}) - 1 + \xi(G_1) - 1
\geq  (k-1)\xi(G_1) - 2(k-2) + \xi(G_1) - 2 = k\xi(G_1) - 2(k-1) = k\xi(H) - 2(k-1)$.
We restate this result as
$\xi(G_k) \geq k(\xi(H) - 2) +2$.

\vspace{.3cm}

Thus $\lim_{ k \rightarrow \infty} \frac{\xi(G_k)}{r(G_k)}  
\geq  \lim_{ k \rightarrow \infty} \frac{k(\xi(H)-2) + 2}{k(r(H) - 2) + 2} 
\geq  \lim_{ k \rightarrow \infty} \frac{kt(r(H) - 2)+2}{k(r(H) - 2) + 2}
= t$, and hence there is no constant $c < t$ such that $\xi(G) \leq cr(G)$ for all $s$-graphs $G$.
\hfill\ {\raisebox{0.8ex}{\framebox{}}}\par\bigskip 

With the same construction as in the proof above, we obtain the following corollary.

\begin{corollary}
 If there is a constant $t >  1$ such that for all $s$-graphs $G$, $\xi(G) \leq tr(G)$, then 
$\xi(K) < tr(K)$, for all class 2 $s$-graphs $K$ (i.e. the bound is not attained).
\end{corollary}

{\bf Proof.} Assume there is a class 2 $s$-graph $K$ such that $\xi(K) = tr(K)$. Take two copies of $K-e$ ($e \in E(K)$ and there is an $s$-coloring of $K$ that leaves $r(K)$ edges uncolored; moreover $e$ is also uncolored) and connect them by adding two edges to obtain the $s$-graph $H$. As in the proof of Theorem \ref{no_constant_bound}, we infer
$r(H) = 2(r(K)-1)$ and $\xi(H) \geq 2(\xi(K)-1)$. Thus 
$\xi(H) \geq 2(\xi(K)-1) = 2tr(K) - 2 = tr(H) + 2t -2 > tr(H)$, since $t > 1$.
\hfill\ {\raisebox{0.8ex}{\framebox{}}}\par\bigskip

Graph $O_1$ of Lemma \ref{basic_module} has resistance 3 and oddness 4. Therefore it immediatly follows
with Theorem \ref{no_constant_bound}:

\begin{corollary} \label{no_constant}
 For all $s \geq 3$, there is no constant $c$ such that $1 \leq c < 2$ and $\xi(G) \leq cr(G)$, for all
$s$-graphs $G$.
\end{corollary}

For $s=2$, $s$-graphs are $s$-colorable, hence $r = \xi=0$. For larger values of $s$ we have by Lemma
\ref{not1_and_basic2} that $\xi(G)=2$ if and only if $r(G) =2$. However, if the resitsance is 
bigger than 3, then it is not clear whether there is a general upper bound for the oddness in terms of the resistance. We state the following problem:

\begin{problem}
For $s \geq 3$: does there exist a function $f$ such that $\xi(G) \leq f(r(G))$ for each $s$-graph $G$. 
\end{problem}

\section{Measures of edge-uncolorabilty for general graphs}

Recall that for a graph $G$, $r_v(G)$ denotes the minimum number of vertices that have to be removed from $G$ to obtain a class 1 graph, that is a graph $H$ with $\chi'(H)=\Delta(H)$, and $r'_v(G)$ is the minimum number of vertices that have to be removed from $G$ to get a graph $K$ with $\chi'(K)\leq \Delta(G)$. It is clear that $r'_v(G) \leq r_v(G)$. The following proposition shows that their difference can be arbitrary big.

\begin{proposition}For each positive integer $k$, there is a graph $G_k$ with $r_v(G_k)-r'_v(G_k)=k$.
\end{proposition}

\begin{proof}For $i=0,1,...,k$ let $H_i$ be a multi-triangle with vertices $a_i, b_i, v_i$. Let $H_0$ be the triangle where every edge has multiplicity 2. 
For $1 \leq i \leq k$, let $H_i$ be the triangle with two edges of multiplicity 2, and one simple edge, say $a_ib_i$. Consider the graph $G_k$ obtained from 
$H_0$, $H_1$,...,$H_k$ by adding the edges $b_0a_1, b_1a_2,\dots, b_{k-1}a_k$. It can be easily verified that $\Delta(G_k)=5$, $r(G_k)=1$, $r'_v(G_k)=1$ and $r_v(G_k)=k+1$.
$\square$
\end{proof}

\begin{theorem}For every graph $G$ the following inequalities hold: $\frac{r(G)}{r_v(G)}\leq \left \lfloor \frac{\Delta(G)}{2}\right \rfloor$ and 
$\frac{r(G)}{r'_v(G)}\leq \left \lfloor \frac{\Delta(G)}{2}\right \rfloor$. Furthermore, the bounds are best possible.
\end{theorem}

\begin{proof}Since $r'_v(G) \leq r_v(G)$, it suffices to prove only the second inequality.

For the proof of this inequality, let us assume that $r'_v(G)=k$ and let $V'=\{v_1,...,v_k\}$ be a set of vertices, such that $\chi'(G-V')\leq \Delta(G)$. For $i=0,...,k$ define the graph $G_i$ as $G_i=G-\{v_{i+1},...,v_k\}$. Note that $G_0=G-V'$ and $G_k=G$.

Also, for a graph $H$ define $r^{\Delta(G)}_{e}(H)$ as the minimum number of edges that should be removed from $H$ to get a $\Delta(G)$-edge-colorable subgraph. Observe that $r^{\Delta(G)}_{e}(G_0)=0$ and $r^{\Delta(G)}_{e}(G)=r(G)$.

Note that since $r^{\Delta(G)}_{e}(G_0)=0$, it suffices to show that for $i=1,...,k$
\begin{equation*}
r^{\Delta(G)}_{e}(G_i)\leq r^{\Delta(G)}_{e}(G_{i-1})+\left \lfloor \frac{\Delta(G_i)}{2}\right \rfloor.
\end{equation*}

Let $\theta_{i-1}$ be a $\Delta(G)$-edge-coloring of $G_{i-1}$ that leaves $r^{\Delta(G)}_{e}(G_{i-1})$ edges uncolored. Assume that $\theta_{i-1}$ uses the colors $\{1,...,\Delta(G)\}$. We will describe an algorithm that from $\theta_{i-1}$ constructs a $\Delta(G)$-edge-coloring $\theta_{i}$ of $G_{i}$ that leaves at most $r^{\Delta(G)}_{e}(G_{i-1})+\left \lfloor \frac{\Delta(G_i)}{2}\right \rfloor$ edges uncolored. Clearly, this will prove the required inequality.

Let $e_1,...,e_l$ be the new edges of $G_i$, i.e. the edges that are incident to $v_i$ in $G_i$. Consider $\theta_{i-1}$ as a $\Delta(G)$-edge-coloring of $G_{i}$, in which the edges $e_1,...,e_l$ are left uncolored and $\bar{C}(v_i)=\{1,...,\Delta(G)\}$, where, for a vertex $v$, $\bar{C}(v)$ denotes the set of colors missing at $v$.

Repeatedly, apply the following rules until none of them is applicable.\\

Rule 1:if there is an uncolored edge $v_iu$ and a color $\alpha \in \bar{C}(v_i)\cap \bar{C}(u)$, then color the edge $v_iu$ by $\alpha $;

Rule 2: if there is an uncolored edge $v_iu$, $\alpha \in \bar{C}(v_i)$, $\beta \in \bar{C}(u)$, such that the $\alpha-\beta$ alternating path starting from vertex $u$ does not come back to $v_i$, then exchange the colors $\alpha$ and $\beta$ on this path, color the edge $v_iu$ by $\alpha $;

Rule 3: if there are uncolored edges $v_iu$, $v_iw$ ($u\neq w$), $\alpha \in \bar{C}(u)\cap \bar{C}(w)$, $\beta \in \bar{C}(v_i)$, $\gamma \in \bar{C}(v_i)$, such that the $\alpha-\beta$-alternating path starting from vertex $u$ and $\alpha-\gamma$-alternating path starting from vertex $w$ share an edge $v_iv$ ($v\neq u,w$) colored by $\alpha$, then color the edge $v_iw$ by $\alpha$, clear the color of the edge $v_iv$ and exchange the colors on the $\alpha-\beta$-alternating path starting from vertex $u$.\\

Let $\theta_{i}$ be the resulting edge-coloring of $G_i$. We claim that there are at most $r^{\Delta(G)}_{e}(G_{i-1})+\left \lfloor \frac{\Delta(G_i)}{2}\right \rfloor$ uncolored edges in $G_i$.

Note that none of the rules 1,2 and 3 changes the number of uncolored edges that belong to $G_{i-1}$. Thus, it suffices to show that there are at most $\left \lfloor \frac{\Delta(G_i)}{2}\right \rfloor$ uncolored edges adjacent to $v_i$ in $G_i$.

Let $u_1,...,u_p$ be the vertices of $G_{i-1}$ that are connected to $v_i$ by an uncolored edge of $\theta_{i}$, and let $e^{(1)}_1,...,e^{(1)}_{k_1}$ be the uncolored edges that connect vertices $v_i$ and $u_1$, $e^{(2)}_{k_1+1},...,e^{(2)}_{k_1+k_2}$ be the ones that connect vertices $v_i$ and $u_2$,..., $e^{(p)}_{k_1+...+k_{p-1}+1},...,e^{(p)}_{k_1+...+k_p}$ be the ones that connect vertices $v_i$ and $u_p$.

Let $\bar{C}(u_j)=\{\beta^{(j)}_{1},...,\beta^{(j)}_{t_j}\}$ for $j=1,...,p$, and also let $\bar{C}(v_i)=\{\alpha_1,...,\alpha_s\}$. Note that $t_j\geq k_j$ and $s\geq k_1+...+k_p$. Moreover, since the rules 1,2 and 3 are not applicable, we have:
\begin{itemize}
	\item for any $\alpha \in \bar{C}(v_i)$ and $\beta \in \bar{C}(u_j)$, the $\alpha-\beta$ alternating path starting from vertex $u_j$ ends at vertex $v_i$, which means that any uncolored edge $e$ connecting vertices $v_i$ and $u_j$ lies in an odd cycle $C^{e}_{\alpha, \beta}$;
	\item $\bar{C}(u_q)\cap \bar{C}(u_r)=\emptyset$ $1\leq q<r\leq p$.
\end{itemize}
Consider the cycles:

\begin{eqnarray*}
C^{e^{(1)}_{1}}_{\alpha_1,\beta^{(1)}_{1}},...,C^{e^{(1)}_{k_1}}_{\alpha_{k_1},\beta^{(1)}_{k_1}},\\
C^{e^{(2)}_{k_1+1}}_{\alpha_{k_1+1},\beta^{(2)}_{1}},...,C^{e^{(2)}_{k_1+k_2}}_{\alpha_{k_1+k_2},\beta^{(2)}_{k_2}},\\
\dots\\
C^{e^{(p)}_{k_1+...+k_{p-1}+1}}_{\alpha_{k_1+...+k_{p-1}+1},\beta^{(p)}_{1}},...,C^{e^{(p)}_{k_1+...+k_{p}}}_{\alpha_{k_1+...+k_p},\beta^{(p)}_{k_p}}.\end{eqnarray*}

These cycles are edge-disjoint in the neighborhood of $v_i$, thus with any uncolored edge adjacent to $v_i$ there is a unique colored edge adjacent to it. This shows that $v_i$ is adjacent to at most $\left \lfloor \frac{d_{G_i}(v_i)}{2} \right \rfloor \leq \left \lfloor \frac{\Delta(G_i)}{2}\right \rfloor$ uncolored edges and completes the proof of the second inequality. \\

It remains to show that the bounds are best possible. We show  that
for each $\Delta \geq 2$ there is a $\Delta$-regular graph $G_{\Delta}$ with $\frac{r(G_{\Delta})}{r_v(G_{\Delta})}=\frac{r(G_{\Delta})}{r'_v(G_{\Delta})}=\left \lfloor \frac{\Delta(G)}{2}\right \rfloor$.

Case 1: $\Delta=2k$. Let $G_{\Delta}$ be the complete graph $K_{2k+1}$. Observe that $\Delta(K_{2k+1})=2k$, $r(K_{2k+1})=k$ and $r_v(K_{2k+1})=r'_v(K_{2k+1})=1$. Thus: $r(G_{\Delta})=r_v(G_{\Delta})\left \lfloor \frac{\Delta(G_{\Delta})}{2} \right \rfloor=r'_v(G_{\Delta})\left \lfloor \frac{\Delta(G_{\Delta})}{2} \right \rfloor$.

Case 2: $\Delta=3$. Let $G_{\Delta}$ be any cubic graph that is not $3$-edge-colorable. Then, as it is shown in \cite{Steffen_1998}, $r(G_{\Delta})=r_v(G_{\Delta})=r'_v(G_{\Delta})$.

Case 3: $\Delta=2k+1\geq 5$. Let $H$ be the graph obtained from $K_{2k+1,2k+1}$ by subdividing one of its edges. Take $k$ copies of the graph $H$ and identify the vertices of degree two to get the graph $G$. Now, take two copies of $G$ and connect the vertices of degree $2k$ by an edge to get the graph $G_{\Delta}$. Then
$\Delta(G_{\Delta})=2k+1$, $r(G_{\Delta})=2k$ and $r_v(G_{\Delta})=r'_v(G_{\Delta})=2$. Thus: $r(G_{\Delta})=r_v(G_{\Delta})\cdot \left \lfloor \frac{\Delta(G_{\Delta})}{2} \right \rfloor=r'_v(G_{\Delta})\cdot \left \lfloor \frac{\Delta(G_{\Delta})}{2}\right \rfloor$.
\hfill $\square$
\end{proof}

\end{document}